\documentclass{acm_proc_article-sp}
\usepackage{times}
\usepackage{epsfig}
\usepackage{graphicx}
\usepackage{amsmath}
\usepackage{amssymb}
\usepackage{multirow}
\usepackage{pbox}
\usepackage{framed}
\usepackage{subcaption}
\usepackage[linesnumbered,ruled,vlined]{algorithm2e}
\usepackage{color}
\usepackage{url}
\usepackage{epstopdf}
\newtheorem{theorem}{Theorem}[section]
\newtheorem{lemma}[theorem]{Lemma}

\newtheorem{definition}{Definition}

\let\oldnl\nl% Store \nl in \oldnl
\newcommand{\nonl}{\renewcommand{\nl}{\let\nl\oldnl}}% Remove line number for one line

\begin{document}

\title{Diverse Yet Efficient Retrieval using Hash Functions}

\numberofauthors{3}
\author{
% 1st. author
\alignauthor
Vidyadhar Rao \\ 
       \affaddr{IIIT Hyderabad, India}\\
       \email{first.last@research.iiit.ac.in} \\
% 2nd. author
\and
\alignauthor
Prateek Jain \\ 
       \affaddr{Microsoft Research India}\\
       \email{prajain@microsoft.com} \\
% 3rd. author
\alignauthor 
C.V Jawahar \\ 
       \affaddr{IIIT Hyderabad, India}\\
       \email{jawahar@iiit.ac.in}\\
  % use '\and' if you need 'another row' of author names
}

\maketitle
\begin{abstract}
Typical retrieval systems have three requirements: a) Accurate retrieval i.e., the method 
should have high precision, b) Diverse retrieval, i.e., the obtained set of points should be 
diverse, c) Retrieval time should be small. However, most of the existing methods address only 
one or two of the above mentioned requirements. In this work, we present a method based on 
{\em randomized} locality sensitive hashing which tries to address all of the above requirements simultaneously. 
While earlier hashing approaches considered approximate retrieval to be acceptable only for the 
sake of efficiency, we argue that one can further exploit approximate retrieval to provide 
impressive trade-offs between accuracy and diversity. We extend our method to the problem of 
multi-label prediction, where the goal is to output a diverse and accurate set of labels for a 
given document in real-time. Moreover, we introduce a new notion to simultaneously evaluate 
a method's performance for both the precision and diversity measures. Finally, we present 
empirical results on several different retrieval tasks and show that our method 
retrieves diverse and accurate images/labels while ensuring $100x$-speed-up over the 
existing diverse retrieval approaches.
\end{abstract}

\category{H.3.3}{Information Search and Retrieval}[Selection Process]
\category{G.1.6}{Optimization}[Quadratic and Integer programming]
\category{G.3}{Probability and statistics}[Probabilistic algorithms]

\terms{Algorithms, Retrieval Performance}


%%%%%%%%% BODY TEXT
\section{Introduction}

Nearest neighbor (NN) retrieval is a critical sub-routine for machine learning, databases, signal processing, and a variety of other disciplines. Basically, we have a database of points, and an input query, the goal is to return the nearest point(s) to the query using some similarity metric. As a na\"{\i}ve linear scan of the database is infeasible in practice, most of the research for NN retrieval has focused on making the retrieval efficient with either novel index structures~\cite{cayton2008learning, yu2011exact} or by approximating the distance computations~\cite{basri2011approximate, jain2010hashing}. That is, the goal of these methods is: a) accurate retrieval, b) fast retrieval. 

However in practice, NN retrieval methods~\cite{ference2013spatial, kucuktunc2013lambda} are expected to meet one more criteria: diversity of retrieved data points. That is, it is typically desirable to find data-points that are diverse and cover a larger area of the space while maintaining high accuracy levels. For instance, when a user is looking for \textit{flowers}, a typical NN retrieval system would tend to return all the images of the same flower (say lilly). But, it would be more useful to show a diverse range of images consisting of \textit{sunflowers}, \textit{lillies}, \textit{roses}, etc. In this work, we propose a simple retrieval scheme that addresses all of the above mentioned requirements, i.e., a) accuracy, b) retrieval time, c) diversity. 

Our algorithm is based on the following simple observation: in most of the cases, one needs to trade-off accuracy for diversity. That is, rather than finding the nearest neighbor, we would need to select a point which is a bit farther from the given query but is {\em dissimilar} to the other retrieved points. Hence, we hypothesis that {\em approximate nearest neighbors} can be used as a proxy to ensure that the retrieved points are diverse. While earlier approaches considered approximate retrieval to be acceptable only for the sake of efficiency, we argue that one can further exploit approximate retrieval to provide impressive trade-offs between accuracy and diversity. %To the extent of our knowledge, this is the first work that shows how to obtain accurate and diverse retrieval using approximate retrieval.  

To this end, we propose a Locality Sensitive Hashing (LSH) based algorithm that guarantees approximate nearest neighbor retrieval in sub-linear time retrieval and superior diversity. We show that the effectiveness of our method depends on {\em randomization} in the design of the hash functions. Further, we modify the standard hash functions to take into account the distribution of the data for better performance. In our approach, it is easy to see that we can obtain higher accuracy with poor diversity and higher diversity with poor accuracy. Therefore, similar to precision and recall, there is a need to balance between accuracy and diversity in the retrieval. We keep a balance between accuracy and diversity and try to maximize the harmonic mean of these two criteria. Our method retrieves points that are sampled uniformly at {\em random} to ensure diversity in the retrieval while maintaining reasonable number of relevant ones. Figure~\ref{fig:toyDataset} contrasts our approach with the different retrieval methods. 

\begin{figure*}
        \centering
        \begin{subfigure}[b]{0.3\textwidth}
        \centering
        			\fbox{
                \includegraphics[width=0.75\textwidth]{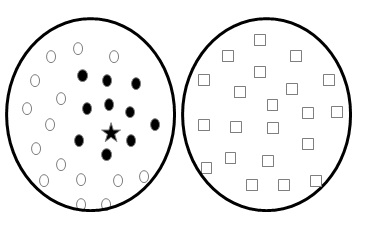}	
                }
                \caption{Accurate, Not diverse}
                \label{fig:nn_app}
        \end{subfigure}
        \begin{subfigure}[b]{0.3\textwidth}
        \centering
        			\fbox{
                \includegraphics[width=0.75\textwidth]{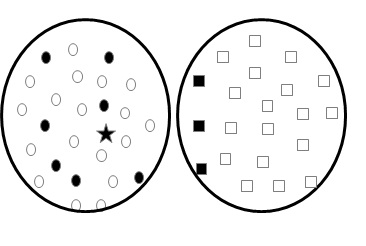}
                	}
                \caption{Not Accurate, diverse}
                \label{fig:greedy_app}
        \end{subfigure}
        \begin{subfigure}[b]{0.3\textwidth}
        \centering
        			\fbox{
                \includegraphics[width=0.75\textwidth]{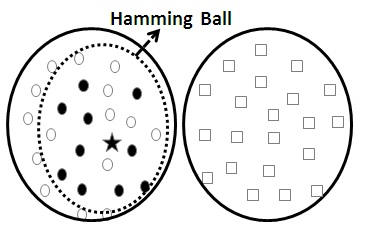}
                	}
                \caption{Accurate, diverse}
                \label{fig:lsh_div_app}
        \end{subfigure}
        \caption{Consider a toy dataset with two classes: class A ($\circ$) and class B ($\square$). We show the query point ($\star$) along with ten points ($\bullet, \blacksquare$) retrieved by various methods. In this case, we consider diversity to be the average pairwise distance between the points. a) A conventional similarity search method (e.g: k-NN) chooses points very close to the query and therefore, shows poor in diversity. b) Greedy methods offer diversity but might make poor choices by retrieving points from the class B. c) Our method finds a large set of approximate nearest neighbors within a hamming ball of a certain radii around the query point and also ensuring the diversity among the points.}
        \label{fig:toyDataset}
\end{figure*}

%\newpage

The main contribution of this paper can be summarized as follows:
\begin{enumerate}
	\item We formally define the diverse retrieval problem and show that in its general form is NP-hard and that also the existing methods are computationally expensive.
	\item While approximate retrieval is acceptable only for sake of efficiency, we argue that one can further exploit approximate retrieval to provide impressive trade-offs between accuracy and diversity.
	\item We propose hash functions that characterizes the locality sensitive hashing to retrieve approximate nearest neighbors in sub-linear time and superior diversity.  
	\item We extend our method to diverse multi-label prediction problem and show that our method is not only orders of magnitude faster than the existing diverse retrieval methods but also produces accurate and diverse set of labels.	
	%Our method retrieves diverse and accurate images/labels/tags while ensuring $100x$-speed-up over the existing diverse retrieval methods.
\end{enumerate}

\section{Related Work}

\subsection{Optimizing Relevance and Diversity}

Many of the diversification approaches are centered around an optimization problem that is derived from both relevance and diversity criteria. These methods can be broadly categorized into the following two approaches: (a) {\em Backward selection}: retrieve all the relevant points and then find a subset among them with high diversity, (b) {\em Forward selection}: retrieve points sequentially by combining the relevance and diversity scores with a greedy algorithm~\cite{carbonell1998use, deselaers2009jointly, he2012gender, zhai2003beyond}. Most popular among these methods is MMR optimization~\cite{carbonell1998use} which recursively builds the result set by choosing the next optimal selection given the previous optimal selections. 

Recent works~\cite{sanner2011diverse, wang2010statistical} have shown that natural forms of diversification arise via optimization of rank-based relevance criteria such as average precision and reciprocal rank. It is conjectured that optimizing $n-call@k$ metric correlates more strongly with diverse retrieval. More specifically, it is theoretically shown~\cite{sanner2011diverse} that greedily optimizing expected $1-call@k$ w.r.t a latent subtopic model of binary relevance leads to a diverse retrieval algorithm that shares many features to the MMR optimization. However, the existing greedy approaches that try to solve the related optimization problem are computationally more expensive than the simple NN, rendering them infeasible for large scale retrieval applications.

Complementary to all the above methods, our work recommands diversity in retrieval using randomization and not optimization. In our work, instead of finding exact nearest neighbors to a query, we retrieve approximate nearest neighbors that are diverse. Intuitively, our work parallels with these works~\cite{sanner2011diverse, wang2010statistical}, and generalizes to arbitrary relevance/similairty function. In our findings, we theoretically show that approximate NN retrieval via locality sensitive hashing naturally retrieve points which are diverse.

\subsection{Application to multi-label prediction}

A typical application of multi-label learning is automatic image/video tagging~\cite{carneiro2007supervised, weston2010large}, where the goal is to tag a given image with all the relevant concepts/labels. Other examples of multi-label instance classification include bid phrase recommendation~\cite{agrawal2013multi}, categorization of Wikipedia articles etc. In all cases, the query is typically an instance (\textit{e.g.}, images, text articles) and the goal is to find the most relevant labels (\textit{e.g.,} objects, topics). Moreover, one would like the labels to be diverse. 

For instance, for a given image, we would like to tag it with a small set of diverse labels rather than several very similar labels. However, the given labels are just some names and we typically do not have any features for the labels. For a given image of a lab, the appropriate tags might be chair, table, carpet, fan etc. In addition to the above requirement of accurate prediction of the positive labels (tags), we also require the obtained set of positive labels (tags) to be {\em diverse}. That is, for an image of a lab, we would prefer tags like $\{$table, fan, carpet$\}$, rather than tags like $\{$long table, short table, chair$\}$. The same problem can be extended to several other tasks like document summarization, wikipedia document categorization etc. Moreover, most of the existing multi-label algorithms run in time linear in the number of labels which renders them infeasible for several real-time tasks ~\cite{WestonBU11, yu2014large}; exceptions include random forest based methods~\cite{agrawal2013multi, fxml}, however, it is not clear how to extend these methods to retrieve diverse set of labels.

In Section~\ref{sec:multi}, we propose a method that extends our diverse NN retrieval based method to obtain diverse and sub-linear (in the number of labels) time multi-label prediction. Our method is based on the LEML method~\cite{yu2014large} which is an embedding based method. The key idea behind embedding based methods for multi-label learning is to embed both the given set of labels as well as the data points into a common low-dimensional space. The relevant labels are then recovered by NN retrieval for the given query point (in the embedded space). That is, we embed each label $i$ into a $k$-dimensional space (say $y_i\in \mathbb{R}^k$) and the given test point is also embedded in the same space (say $x_q\in \mathbb{R}^k$). The relevant labels are obtained by finding $y_i$'s that are closest to $x_q$. Note that as the final prediction reduces to just NN retrieval, we can apply our method to obtain diverse set of labels in sub-linear time. 

\subsection{Evaluation Measures}
The need for diversity is not limited to retrieval and there has been significant research in many applications~\cite{dagli2006utilizing, khan2013efficient, ntoutsi2014strength}. In practice, diversity is a subjective phenomenon~\cite{mazur2010cultural}. For example, in active learning~\cite{dagli2006utilizing}, a diversity measure based on Shannon`s entropy is used. Probabilistic models like determinental point processes~\cite{gillenwater2012near, kulesza2011k} evaluate the diversity using real human feedback via Amazon`s Mechanical Turk. Structured SVM based framework~\cite{yue2008predicting} measures diversity using subtopic coverage on manually labelled data. 

Thus, the evaluation measures used to assess the performance of different methods are also different. In our work, the definition of what constitutes diversity varies across each task and is clearly described. As mentioned above, we use harmonic mean between accuracy and diversity as the main performance measure. We believe that this performance measure is suitable for several applications and helps us empirically compare different methods.

{\bf Paper Organization}: First, we formalize the diverse retrieval problem in Section~\ref{sec:prob_form}. We then present our diverse retrieval methods based on locality senstive hash functions in Section~\ref{sec:method}. We also present diverse multi-label prediction method in Section~\ref{sec:multi}. We describe our performance measure and experimental setup in Section~\ref{sec:expset}. Then, in Section~\ref{sec:exps}, we provide empirical results on two different (image and text) applications. Finally, we present our conclusions in Section~\ref{sec:conc}.

%%%

\section{Diverse Retrieval Optimization}
\label{sec:prob_form}

Given a set of data points ${\cal X}=\{(x_1, y_1), \dots, (x_n, y_n)\}$ where $x_i \in \mathbb{R}^d$, $y_i$ is a label and a query point $q\in \mathbb{R}^d$, the goal is two-fold: a) retrieve a set of points ${\cal R}_q=\{x_{i_1}, \dots, x_{i_k}\}$ such that a majority of their labels correctly predicts the label of $q$. b) The set of retrieved points ${\cal R}_q$ is ``diverse''.  Note that, in this work we are only interested in finding $k$ points that are relevent to the query. We formally start with the two definitions that are empirically successful and are widely used measures for similarity and diversity in the context of retrieval: 
\begin{definition}
	 For a given two points, {\bf dis-similarity} is defined as the distance between the two points, say $x$ and $y$, i.e., 
		$DisSim(x, y)=\|x-y\|_2^2$
	\label{def:sim}
\end{definition}
\begin{definition}
	 For a given set of points, {\bf diversity} is defined as the average pairwise distance between the points of the set, i.e., 
		$Div({\cal R}_q)=\sum_{a, b}\|x_{i_a}-x_{i_b}\|_2^2$
	\label{def:div}
\end{definition}
With the above definitions, our goal is to find a subset of $k$ points which are both relevent to the query and diversified among themselves. Although it is not quite clear on how relevance and diversity should be combined, we adopt a reminiscent~\cite{lin2011class} of the general paradigm in machine learning of combining loss functions that measures quality(e.g., training error, prior, or ``relevance'') and a regularization term that encourages desirable properties (e.g. smoothness, sparsity, or ``diversity''). To this end, we define the following optimization problem. 
\begin{equation}
\label{eq:qip_obj}
\begin{aligned}
  \min  		&\hspace{0.5em} \lambda \Sigma_{i=1}^{n}\alpha_{i}\|q - x_{i}\|^{2} - (1-\lambda)\Sigma_{ij}\alpha_{i}\alpha_{j}\|x_{i} - x_{j}\|^{2} \\
  \mathrm{s.t.}   &\hspace{0.5em} \Sigma_{i=1}^{n}\alpha_{i} = k; \forall i \in \{1, \ldots n\} \alpha_{i} \in \{0,1\}  \\
\end{aligned} 
\end{equation}

where $\lambda \in [0, 1]$ is a parameter that defines the trade-off between the two terms, and $\alpha_{i}$ takes the value 1 if $x_{i}$ is present in the result and 0 if it is not included in the retrieved result. Without loss of generality, we assume that $x_{i}, q$ are normalized to unit norm, and with some simple substitutions like $\alpha = [\alpha_{1}, \ldots \alpha_{n}]$, $c = -[q^{T}x_{1}, \ldots, q^{T}x_{n}]$, $G$ be gram matrix with $G_{ij} = x_{i}^{T}x_{j}$, the above objective is equivalent to
\begin{equation}
\label{eq:qip_obj_quad}
\begin{aligned}
  \min  		&\hspace{0.5em} \lambda c^{T}\alpha + \alpha^{T}G\alpha  \\
  \mathrm{s.t.}   &\hspace{0.5em} \alpha^{T}1 = k; \alpha \in \{0, 1\}^{n}\\
\end{aligned} 
\end{equation}
From now on, we refer to the diverse retrieval problem in the form of the optimization problem in Eq.(\ref{eq:qip_obj_quad}). Finding optimal solutions for the quadratic integer program in Eq.(\ref{eq:qip_obj_quad}) is NP-hard~\cite{wolsey1998integer}. Usually QP relaxations~\cite{ravikumar2006quadratic, jancsary2013learning} (which are often called linear relaxations), where integer constraints are relaxed to interval constraints, are efficiently solvable.
\begin{equation}
\label{eq:qip_obj_quad_relax}
\begin{aligned}
  \min  		&\hspace{0.5em} \lambda c^{T}\alpha + \alpha^{T}G\alpha  \\
  \mathrm{s.t.}   &\hspace{0.5em} \alpha^{T}1 = k; 0 \le \alpha \le 1\\
\end{aligned} 
\end{equation}

In this work, we consider the following simple approach\footnote{We refer to this method with QP-Rel in our experimental evaluations as one of our baselines.} to solve Eq.(\ref{eq:qip_obj_quad}): We first remove the integrality constraint on the variables i.e., allow variables to take on non-integral values to obtain a quadratic optimization program in Eq.(\ref{eq:qip_obj_quad_relax}). Now, we find the optimal solution to the quadratic program in Eq.(\ref{eq:qip_obj_quad_relax}). Note that the optimal solution to the relaxed problem is not necessarily integral. Therefore, we select the top $k$ values from the fractional solution and report it as the integral feasible solution to Eq.(\ref{eq:qip_obj_quad}). Although, this method yields a good solution to Eq.(\ref{eq:qip_obj_quad}) i.e., obtains accurate and diverse retrieval, solving the QP Relaxation is much more time consuming than the existing solutions (\textit{see Table~\ref{table:imagenet_results} for more details}). Therefore, it is of greatest interest to look for computationally efficient solutions for the diverse retrieval problem.

To this end, the existing approaches i.e., greedy methods~\cite{carbonell1998use, deselaers2009jointly, he2012gender, zhai2003beyond} for Eq.(\ref{eq:qip_obj}) and the QP relaxation method for Eq.(\ref{eq:qip_obj_quad}) suffer from two drawbacks: a) Running time of the algorithms is very high as it is required to recover several exact nearest neighbors. b) The obtained points might all be from a very small region of the space and hence the diversity of the selected set might not be large. c) Computation of the gram matrix may require an unreasonably large amount of memory overhead for large datasets. In this work, we propose a simple approach to overcome the above three issues.

\section{Methodology}
\label{sec:method}

To find nearest neighbors, the basic LSH algorithm concatenates a number of functions $h\in{\cal H}$ into one hash function $g\in{\cal G}$. Informally, we say that ${\cal H}$ is {\em locality-sensitive} if for any two points $a$ and $b$, the probability of $a$ and $b$ collide under a random choice of hash function depends only on the distance between $a$ and $b$. Several such families are known in the literature, see ~\cite{andoni2006near} for an overview.

\begin{definition} (Locality-sensitive hashing): A family of hash functions ${\cal H}:R^{d} \rightarrow \{0,1\} $ is called ($r, \epsilon, p, q$ )-sensitive if for any $a, b \in R^{d}$
\[
\begin{cases}
    Pr_{h \in {\cal H}}[h(a) = h(b)] \geq p, & \text{if } d(a, b) \leq r \\
    Pr_{h \in {\cal H}}[h(a) = h(b)] \leq q, & \text{if } d(a, b) \geq (1+\epsilon)r
\end{cases}
\]	 
Here, $\epsilon>0$ is an arbitrary constant, $p> q$ and $d(.,.)$ is some distance function.
\end{definition}

In this work, we use $\ell_2$ norm as the distance function and adopt the following hash function: 
\begin{equation}
	h(a)=sign(r\cdot a)
\end{equation}

where $r\sim {\cal N}(0, I)$. It is well known that $h(a)$ is a LSH function w.r.t $\ell_2$ norm and it is shown to satisfy the following: 
\begin{equation}
  \label{eq:hash}
  Pr(h(a)\neq h(b))=\frac{1}{\pi}\cos^{-1}\left(\frac{a\cdot b}{\|a\|_2\|b\|_2}\right).
\end{equation}

Our approach is based on the following high-level idea: perform {\em randomized approximate} nearest neighbor search for $q$ which selects points randomly from a small disk around $q$. As we show later, locality sensitive hashing with standard hash functions actually possess such a quality. Hence, the retrieved set would not only be accurate (i.e. has small distance to $q$) but also diverse as the points are selected randomly from the neighborhood of $q$. In our algorithm, we retrieve more than the required $k$ neighbors and then select a set of diverse neighbors by using a greedy method. See Algorithm~\ref{algo:alg_lsh} for a detailed description of our approach.

\begin{algorithm}
\DontPrintSemicolon 
\KwIn{${\cal X}=\{x_1\dots, x_n\}$, where $x_i \in R^{d}$, a query $q\in \mathbb{R}^d$ and $k$ an integer.}
 
\textbf{Preprocessing}: For each $i \in [1 \ldots L]$, construct a hash function, $g_{i} =[h_{1, i},\ldots ,h_{l, i}]$, where $h_{1, i}, \ldots ,h_{l, i}$ are chosen at random from $\cal H$. Hash all points in ${\cal X}$ to the $i^{th}$ hash table using the function $g_{i}$ \;
$R \leftarrow \phi $\;
\For{$i \gets 1$ \textbf{to} $L$} {
   Perform a hash of the query $g_{i}(q)$ \;
   Retrieve points from $i^{th}$ hash table \& append to ${\cal R}_q$\;
}
${\cal S}_{q} \leftarrow \phi $\;	
\For{$i \gets 1$ \textbf{to} $k$}{

	$r^{*} \leftarrow \text{argmin}_{(r \in {\cal R}_q)} (\lambda\|q - r\|^{2} - \frac{1}{i} \Sigma_{s \in {\cal S}_{q}}\|r - s\|^{2})$\;
	${\cal R}_q \leftarrow  {\cal R}_q \setminus r^{*}$\;
	${\cal S}_{q} \leftarrow {\cal S}_{q} \cup r^{*}$\;

%	$r \leftarrow \alpha || ||_{2} + \frac{1-\alpha}{i}&nbsp;\Sigma_{j=1}^{i}D_s(x_{r_{i}}, x_{r_{j}}) $
}
\KwOut{${\cal S}_{q}$, $k$ diverse set of points}
\caption{LSH with random hash functions (LSH-Div)}
\label{algo:alg_lsh}
\end{algorithm}

The algorithm executes in two phases: i) perform search through the hash tables, line(2-4), to report the approximate nearest neighbors, $R_{q} \subset {\cal X}$ and ii) perform $k$ iterations, line(6-9), to report a diverse set of points, $S_{q} \subset R_{q}$. Throughout the algorithm, several variables are used to maintain the trade-off between the accuracy and diversity of the retrieved points. The essential control variables that direct the behaviour of the algorithm are: i) the number of points retrieved from hashing, $|R_{q}|$ and ii) the number of diverse set of points to be reported, $k$. Here, $R_{q}$ can be controlled at the design of hash function, i.e., the number of matches to the query is proportional to $n^{\frac{1}{1+\epsilon}}$. Therefore, line 7 (can be optional) is critical for the efficiency of the algorithm, since it is an expensive computation, especially when $|R_{q}|$ is very big, or $k$ is large. More details of our algorithm are discussed in section~\ref{sec:alg_anlys}.

\subsection{Diversity in Randomized Hashing} 
\label{sec:lsh_div}

An interesting aspect of the above mentioned LSH function in Eq.(\ref{eq:hash}) is that it is unbiased towards any particular direction, i.e., $Pr(h(q)\neq h(a))$ is dependent only on $\|q-a\|_2$ (assuming $q, a$ are both normalized to unit norm vectors). But, depending on a sample hyper-plane $r\in \mathbb{R}^d$, a hash function can be biased towards one or the other direction, hence preferring points from a particular region. Interestingly, we show that if the number of hash bits is large, then all the directions are sampled uniformly and hence the retrieved points are sampled uniformly from all the directions. {\em That is, the retrieval is not biased towards any particular region of the space.} We formalize the above observation in the following lemma. 
\begin{definition}({\bf Hoeffding`s Inequality}~\cite{lugosi2004concentration})
	Let $Z_{1}, \ldots, Z_{n}$ be $n$ i.i.d. random variables with $f(Z) \in [a, b]$ . Then for all $\epsilon \ge 0$, with probability at least $1-\delta$ we have 
	\begin{equation*}
		P[\| \frac{1}{n}\sum_{i=1}^{n}f(Z_{i}) - E(f(Z)) \|] \le (b-a)\sqrt{\frac{\log(\frac{2}{\delta})}{2n}}
	\end{equation*}
\end{definition}

\begin{lemma}
\label{lm:lsh_lemma}
    Let $q\in \mathbb{R}^d$ and let ${\cal X}_q=\{x_1, \dots, x_m\}$ be unit vectors such that $\|q-x_i\|_2=\|q-x_j\|_2=r,$ $\forall i,j$. Let $p=\frac{1}{\pi}\cos^{-1}(1-r^2/2)$. Also, let $r_1, \dots, r_\ell \sim {\cal N}(0, I)$ be $\ell$ random vectors. Define hash bits $g(x)=[h_1(x) \dots h_\ell(x)]\in \{0, 1\}^{1\times \ell}$, where hash functions $h_b(x)=sign(r_b\cdot x),$ $1\leq b\leq \ell$. Then, the following holds $\forall i$: 
%$$p-\frac{1}{\sqrt{lp}}\leq \frac{1}{l}\|g(q)-g(x_i)\|_1\leq p+\frac{1}{\sqrt{lp}},$$ 
\begin{equation*}
	p-\sqrt{\frac{\log(\frac{2}{\delta})}{2l}} \le \frac{1}{l} ||g(q) -g(x_{i})||_{1} \le p+\sqrt{\frac{\log(\frac{2}{\delta})}{2l}}
\end{equation*}

That is, if $\sqrt{l}\gg 1/p$, then hash-bits of the query $q$ are almost equi-distant to the hash-bits of each $x_i$.  
\end{lemma}
\begin{proof}
Consider random variable $Z_{ib},\ 1\leq i\leq m, \ 1\leq b\leq \ell$ where $Z_{ib}=1$ if $h_b(q)\neq h_b(x_i)$ and $0$ otherwise. Note that $Z_{ib}$ is a Bernoulli random variable with probability $p$. Also, $Z_{ib},\ \forall 1\leq b\leq \ell$ are all independent for a fixed $i$. Hence, applying Hoeffding's inequality, we obtain the required result. 
\end{proof}
Note that the above lemma shows that if $x_1, \dots, x_m$ are all at distance $r$ from a given query $q$ then their respective hash bits are also at a similar distance to the hash bits of $q$. {\em That is, assuming randomization selection of the candidates from a hash bucket, probability of selecting any $x_i$ is almost the same. That is, the points selected by LSH are nearly uniformly at random and are diverse.} 

\subsection{Randomized Compact Hashing}

In Algorithm~\ref{algo:alg_lsh}, we obtained hash functions by selecting hyper-planes from a normal distribution. The conventional LSH approach considers only random projections. Naturally, by doing random projection, we will lose some accuracy. But we can easily fix this problem by doing multiple rounds of random projections. However, we need to perform a large number of projections (i.e. hash functions in the LSH setting) to increase the probability that similar points are mapped to similar hash codes. A fundamental result of Johnson and Lindenstrauss Theorem~\cite{john1984} says that $O(\frac{\ln{n}}{\epsilon^{2}})$ random projections are needed to preserve the distance between any two pair of points, where $\epsilon$ is the relative error.

Therefore, using many random vectors to generate the hash tables (a long codeword), leads to a large storage space and a high computational cost, which would slow down the retrieval procedure. In practice, however, the data lies in a very small dimensional subspace of the ambient dimension and hence a random hyper-plane may not be very informative. Instead, we wish to use more data driven hyper-planes that are more discriminative and separate out neighbors from far-away points. To this end, we obtain the hyper-planes $r$ using principal components of the given data matrix. Principal components are the directions of highest variance of the data and captures the geometry of the dataset accurately. Hence, by using principal components, we hope to reduce the required number of hash bits and hash tables required to obtain the same accuracy in retrieval. 

That is, given a data matrix $X\in \mathbb{R}^{d\times n}$ where $i$-th column of $X$ is given by $x_i$, we obtain top-$\alpha$ principal components of $X$ using SVD. That is, let $U\in \mathbb{R}^{d\times \alpha}$ be the singular vectors corresponding to the top-$\alpha$ singular values of $X$. Then, a hash function is given by: 
$h(x)=sign(r^TU^Tx)$ where $r\sim {\cal N}(0, I)$ is a random $\alpha$-dimensional hyper-plane. In the subsequent sections, we denote this algorithm using LSH-SDiv.

Many learning based hashing methods~\cite{kulis2009learning, wang2010sequential, weiss2009spectral} are proposed in literature. The simplest of all such approaches is PCA Hashing~\cite{wang2006annosearch} which chooses the random projections to be the principal directions of the data directly. Our algorithm LSH-SDiv method is different from PCA Hashing in the sense that we still select random directions in the top components. Note that the above hash function has reduced randomness but still preserves the discriminative power by projecting the randomness onto top principal components of $X$. As shown in Section~\ref{sec:exps}, the above hash function provides better nearest neighbor retrieval while recovering more diverse set of neighbors. 

\subsection{Diverse Multi-label Prediction}
\label{sec:multi} 

We now present an extension of our method to the problem of multi-label classification. Let ${\cal X}=\{x_1, \dots, x_n\}$, $x_i\in \mathbb{R}^d$ and ${\cal Y}=\{y_1, \dots, y_n\}$, where $y_i \in \{-1, 1\}^L$ be $L$ labels associated with the $i$-th data point. Then, the goal in the standard multilabel learning problem is to predict the label vector $y_q$ accurately for a given query point $q$. Moreover, in practice, the number of labels $L$ is very large, so we require our prediction time to scale sublinearly with $L$. 

In this work, we build upon the LEML method proposed by~\cite{yu2014large} that can solve multi-label problems with a large number of labels and data points. In particular, LEML learns matrices $W, H$ s.t. given a point $q$, its predicted labels is given by $y_q=sign(WH^Tx)$ where $W\in \mathbb{R}^{L\times k}$ and $H\in R^{d\times k}$ and $k$ is the rank of the parameter matrix $WH^T$. Typically, $k\ll \min(d, L)$ and hence the method scales linearly in both $d$ and $L$. For instance, its prediction time is given by $O((d+L)\cdot k)$. 

However, for several widespread problems, the $O(L)$ prediction time is quite large and makes the method infeasible in practice. Moreover, the obtained labels from this algorithm can all be very highly correlated and might not provide a diverse set of labels which we desire. 

We overcome both of the above limitations of the algorithm using the LSH based algorithm introduced in the previous section. We now describe our method in detail. Let $W_1, W_2, \dots, W_L$ be $L$ data points where $W_i\in \mathbb{R}^{1\times k}$ is the $i$-th row of $W$. Also, let $H^Tx$ be a query point for a given $x$. Note that the task of obtaining $\alpha$ positive labels for given $x$ is equivalent to finding $\alpha$ largest $W_i\cdot (H^Tx)$. Hence, the problem is the same as nearest neighbor search with diversity where the data points are given by ${\cal W}=\{W_1, W_2, \dots, W_L\}$ and the query point is given by $q=H^Tx$. 

We now apply our LSH based methods to the above setting to obtain a ``diverse'' set of labels for the given data point $x$. Moreover, the LSH Theorem by~\cite{lsh} shows that the time of retrieval is sublinear in $L$ which is necessary for the approach to scale to a large number of examples. See Algorithm~\ref{algo:alg_mult} for the pseudo-code of our approach. 
\begin{algorithm}[t!]
\KwIn{{\small Train data: ${\cal X}=\{x_1, \ldots, x_n\}$, ${\cal Y}=\{y_1, \dots, y_n\}$. Test data: ${\cal Q}=\{q_1, \dots, q_m\}$}. Parameters: $\alpha$, $k$.}
\nonl 
[W, H]=LEML(${\cal X}$, ${\cal Y}$, $k$)\;
\nonl
${\cal S}_{q}=\text{LSH-SDiv}(W, H^Tq, \alpha),\ \ \forall q\in {\cal Q}$\;
\nonl
$\widehat{y}_q=\text{Majority}(\{y_i \text{ s.t. } x_i\in {\cal S}_q\}), \ \ \ \forall q\in {\cal Q}$\;%	S = LSH-Div( $W$, $H^{T}x$, $\alpha$) \;
	%\tcc*{Preprocessing is done only once.}
\KwOut{${\cal \widehat{Y}}_Q=\{\widehat{y}_{q_1}, \dots, \widehat{y}_{q_m}\}$}
\caption{LSH based Multi-label Classification}
\label{algo:alg_mult}
\end{algorithm}

\subsection{Algorithmic Analysis}
\label{sec:alg_anlys}

As discussed above, locality sensitive hashing is a sub-linear time algorithm for approximate near(est) neighbor search that works by using a carefully selected hash function that causes objects or documents that are similar to have a high probability of colliding in a hash bucket. Like most indexing strategies, LSH consists of two phases: {\em hash generattion}, where the hash tables are constructed and {\em querying}, where the hash tables are used to look up for points similar to the query. Here, we briefly comment on the algorithmic and statistical aspects which are important for the suggested algorithms in the previous sections. %We also mention about the theoretical gurantees where ever applicable. 

{\bf Hash Generation}: In our algorithm, for $l$ specified later, we use a family $\cal{G}$ of hash functions $g(x) = (h_{1}(x), \ldots, h_{l}(x))$, where $h_{i} \in H$. For an integer $L$, the algorithm chooses $L$ functions $g_{1}, \ldots, g_{L}$ from $\cal{G}$, independently and uniformly at random. The algorithm then creates $L$ hash arrays, one for each function $g_{j}$. During preprocessing, the algorithm stores each data point $x \in {\cal X}$ into bucket $g_{j}(x)$ for all $j=1, \ldots, L$. Since the total number of buckets may be large, the algorithm retains only the non-empty buckets by resorting to standard hashing.

{\bf Querying}: To answer a query $q$, the algorithm evaluates $g_{1}, \ldots, g_{L}$, and looks up the points stored in those buckets in the respective hash arrays. For each point $p$ found in any of the buckets, the algorithm computes the distance from $q$ to $p$, and reports the point $p$ if the distance is at most $r$. Different strategies can be adopted to limit the number of points reported to the query $q$, see ~\cite{andoni2006near} for an overview. 

{\bf Accuracy}: Since, the data structure used by LSH scheme is randomized: the algorithm must output all points within the distance $r$ from $q$, and can also output some points within the distance $(1+\epsilon)r$ from $q$. The algorithm guarantees that each point within the distance $r$ from $q$ is reported with a constant (tunable) probability. The parameters $l$ and $L$ are chosen~\cite{indyk1998approximate} to satisfy the requirement that a near neighbors are reported with a probability at least $(1- \delta)$. Note that the correctness probability is defined over the random bits selected by the algorithm, and we do not make any probabilistic assumptions about the data distribution.  

{\bf Diversity}: In lemma~\ref{lm:lsh_lemma}, if the number of hash bits is large i.e,  if $\sqrt{l}\gg 1/p$, then hash-bits of the query $q$ are almost equi-distant to the hash-bits of each point in $x_i$. Then all the directions are sampled uniformly and hence the retrieved points are uniformly spread in all the directions.  Therefore, for reasonable choice of the parameter $l$, the proposed algorithm obtains diverse set of points, ${\cal S}_{q}$ and has strong probabilistic guarantees for large databases of arbitrary dimensions. 

{\bf Scalability}: The time for evaluating the $g_{i}$ functions for a query point $q$ is $O(dlL)$ in general. For the angular hash functions chosen in our algorithn, each of the $l$ bits output by a hash function $g_{i}$ involves computing a dot product of the input vector with a random vector defining a hyperplane. Each dot product can be computed in time proportional to the number of non-zeros $\zeta$ rather than $d$. Thus, the total time is $O(\zeta lL)$. For an interested reader, see that the Theorem 2 of~\cite{charikar2002similarity} guarantees that $L$ is at most $O(N^{\frac{1}{(1 + \epsilon)}})$, where $N$ denotes the total number of points in the database.

%%% 

\section{Experimental Setup}
\label{sec:expset}

We demonstrate our approach applied to the following two tasks: (a) Image Category Retrieval and (b) Multi-label Prediction %and (c) Image Tagging. 

In the case of image retrieval task, we are interested in retrieving diverse images of a specific category. In our case, each of the image categories have associated subcategories (e.g.. {\em flower} is a category and {\em lilly} is a subcategory) and we would like to retrieve the relevant (to the category) but diverse images that belong to different sub-categories. The query is represented as a hyperplane that is trained (SVM~\cite{shalev2011pegasos}) offline to discriminate  between positive and negative classes.  

Next, we apply our diverse retrieval method to the  multi label classification problem; see previous section for more details. Our approach is evaluated on LSHTC\footnote{\url{http://lshtc.iit.demokritos.gr/LSHTC3_CALL}} dataset containing Wikipedia text documents. Each document is represented with the help of a set of categories or class labels. A document can have multiple labels and we are interested in predicting a set of categories to a given document. We model this problem as retrieving a relevant set of labels from a large pool of labels.  In this case, we retrieve labels that match the semantics of the document and also have enough diversity among them.
	 
\subsection{Evaluation Criteria}

In both these experiments, our goal is two-fold: 1) improve diversity in the retrieval and 2) demonstrate speedups of the our proposed algorithms. We now present formal metrics to measure performance of our method on three key aspects of NN retrieval: (i) accuracy (ii) diversity and (iii) efficiency.  We characterize the performance in terms of the following measures: 

\begin{itemize}
 	\item \textbf{Accuracy:} We denote precision at $k$ ($P$@$k$) as the measure of accuracy of the retrieval. This is  the proportion of the relevant instances in the top $k$ retrieved results. In our results, we also report the recall and f-score results when applicable, to compare the methods in terms of multiple measures.

	\item \textbf{Diversity:} For image retrieval, the diversity in the retrieved images is measured using entropy as $D = \frac{\Sigma_{i=1}^{m} s_{i}\log s_{i}}{\log m}$, where $s_{i}$ is the fraction of images of $i^{th}$ subcategory, and $m$ is the number of subcategories for the category of interest. For multi label classification, the relationships between the labels is not a simple tree. It is better captured using a graph and the diversity is then computed using drank~\cite{SwamyRS14}. Drank captures the extent to which the labels of the documents belong to multiple categories.  
	
\item \textbf{Efficiency:} Given a query, we consider retrieval time to be the time between posing a query and retrieving images/labels from the database. For LSH based methods, we first load all the LSH hash tables of the database into the main memory and then retrieve images/labels from the database.  Since, the hash tables are processed offline, we do not consider the time spent to load the hash tables into the retrieval time. All the retrieval times are based on a Linux machine with Intel E5-2640 processor(s) with 96GB RAM.

\end{itemize}

\subsection{Combining Accuracy and Diversity}  

Tradeoffs between accuracy and efficiency in NN retrieval have been studied well in the past~\cite{basri2011approximate, jain2010hashing, rastegari2011scalable, yu2011exact}. Many methods compromise on the accuracy for better efficiency. Similarly, emphasizing higher diversity may also lead to poor accuracy and hence, we want to formalize a metric that captures the trade-off between diversity and accuracy. 

To this end, we use (per data point) harmonic mean of accuracy and diversity as overall score for a given method (similar to f-score providing a trade off between precision and recall). That is, $h-score({\cal A})=\sum_i \frac{2\cdot Acc(x_i)\cdot Diversity(x_i)}{Acc(x_i)+Diversity(x_i)}$, where ${\cal A}$ is a given algorithm and $x_i$'s are given test points. In all of our experiments, parameters are chosen by cross validation such that the overall $h$-score is maximized.

%%% 

\section{Empirical Results}
\label{sec:exps}
%------------------------------------------------------------------------
\renewcommand{\tabcolsep}{3pt}
%------------------------------------------------------------------------

\begin{table*}[!t]
\begin{center}
\caption{We show the performance of various diverse retrieval methods on the ImageNet dataset. We evaluate the performance in terms of precision(P), sub-topic recall(SR) and Diversity(D) measures at top-10, top-20 and top-30 retrieved images. Numbers in \textbf{bold} indicate the top performers. {\bf NH} \textit{corresponds to the method without using any hash function.} Notice that for all methods, \textit{except Greedy}, LSH-Div and LSH-SDiv hash functions consistently show better performance in terms of h-score than the method with NH. Interestingly, we also have the top performers best in terms of retrieval time.} 
\label{table:imagenet_results}
  \begin{tabular}{|l|l||c|c|c|c|r||c|c|c|c|r||c|c|c|c|r|c|c|c|c|r|}
    \hline 
		&{ } & \multicolumn{5}{|c}{\textbf{precision at 10}} & \multicolumn{5}{|c|}{\textbf{precision at 20}} & \multicolumn{5}{|c|}{\textbf{precision at 30}} \\

		\hline
		\textbf{Method} & \textbf{Hash Function} & P & SR & D & h & \pbox{10cm}{time\\(sec)} & P & SR & D & h & \pbox{10cm}{time\\(sec)} & P & SR & D & h & \pbox{10cm}{ time\\(sec)}\\
		\hline \hline
	{ } & NH 			&1.00	&0.60	&0.53 	&0.66	&0.621 				&0.99	&0.72	&0.60 	&0.73	&0.721  &0.99	&0.79	&0.65 	&0.77 &0.845\\ 
	%&0.99	&0.47	&0.70 	&0.57	&0.570
	NN & LSH-Div 	&0.97	&0.79	&0.76 	&\textbf{0.84}&\textbf{0.112}		&0.93	&0.93	&0.86 	&\textbf{0.89}	&\textbf{0.137}  &0.89	&0.98	&0.91 	&0.90 &0.179\\ 
	%&0.98	&0.58	&0.93 	&\textbf{0.70}	&0.096		
	 { }  & LSH-SDiv		&0.98	&0.76	&0.73 	&0.83	&0.181			&0.95	&0.89	&0.85 	&0.89	&0.183	&0.92	&0.95	&0.89 	&\textbf{0.90} &\textbf{0.106} \\ 
	%&0.98	&0.55	&0.89 	&0.68	&0.118				
		\hline \hline
	{ } & NH 			&1.00 	&0.73	&0.69	&0.81	&0.804				&0.99	&0.79 	&0.70 	&0.81	&0.793 &0.99	&0.88	&0.77 	&0.86 &0.901\\ 
	%&1.00 	&0.61 	&1.01	&\textbf{0.76}	&0.763		
	Rerank & LSH-Div 	&0.93	&0.80	&0.76 	&0.83	&0.142				&0.92	&0.93	&0.86 	&0.88	&0.146	&0.87	&0.98	&0.90 	&0.88 &0.214\\ 
	%&0.91	&0.60	&0.99 	&0.72	&0.148				
	{ } & LSH-SDiv 	&0.95	&0.79	&0.76 	&\textbf{0.84}	&\textbf{0.154}		&0.94	&0.91	&0.85 	&\textbf{0.89}	&\textbf{0.179}	&0.90	&0.95	&0.88 	&\textbf{0.89} &\textbf{0.203}\\ 
	%&0.93	&0.60	&0.98 	&0.72	&0.145
		\hline \hline
	{ } & NH    &0.95	&0.75	&0.71 	&0.80	&5.686	&0.98	&0.86	&0.77 	&\textbf{0.85}	&11.193	 &0.97	&0.90	&0.80 	&\textbf{0.87} &17.162\\ 
	%&0.90	&0.58	&0.96 	&\textbf{0.69} 	&2.942 		
	Greedy~\cite{deselaers2009jointly} & LSH-Div 	&0.89	&0.80	&0.76 	&0.81	&1.265				&0.68	&0.88	&0.81 	&0.72	&2.392	&0.53	&0.89	&0.80 	&0.62 &4.437\\ 
	%&0.83	&0.57	&0.93 	&0.66	&0.594				
	{ } & LSH-SDiv	&0.91	&0.78	&0.76 	&\textbf{0.82}	&\textbf{0.986}		&0.69	&0.88	&0.80 	&0.73	&2.417	&0.52	&0.88	&0.80 	&0.61 &3.537\\ 
	%&0.83	&0.56	&0.93 	&0.67	&0.745				
		\hline \hline
	{ }& NH  &0.92	&0.73	&0.68 	&0.77	&5.168	&0.95	&0.86	&0.75 	&0.83	&10.585	 &0.96	&0.90	&0.76 	&0.84  &16.524\\ 
	%&0.83	&0.54	&0.86 	&0.62	&2.417				
	MMR~\cite{carbonell1998use} & LSH-Div 		&0.91	&0.77	&0.73 	&0.80	&1.135				&0.91	&0.92	&0.85 	&0.87	&2.378	 &0.87	&0.97	&0.89 	&0.88	&3.828\\ 
	%&0.86	&0.55	&0.87 	&\textbf{0.64}	&0.648		
	{ } & LSH-SDiv		&0.92	&0.78	&0.75 	&\textbf{0.81} &\textbf{1.102}		&0.93	&0.91	&0.84 	&\textbf{0.88}	&\textbf{2.085} &0.89	&0.96	&0.88 	&\textbf{0.88} &\textbf{4.106}\\ 
	%&0.85	&0.54	&0.86 	&0.63	&0.977
		\hline \hline
	{ } & NH				&1.00	&0.74	&0.69 	&0.81	&\underline{704.9}				&1.00	&0.82	&0.73 	&0.84	&\underline{947.09}	 	&1.00	&0.87	&0.76 	&0.86  &\underline{1137.19}\\ 
	%&0.83	&0.54	&0.86 	&0.62	&2.417				
	QP-Rel & LSH-Div 		&0.93	&0.80	&0.77 	&0.83	&0.487				&0.92	&0.94	&0.86 	&0.88	&0.499	 	&0.86	&0.98	&0.90 	&0.88	&0.502\\ 
	%&0.86	&0.55	&0.87 	&\textbf{0.64}	&0.648		
	{ } & LSH-SDiv		&0.97	&0.78	&0.74 	&\textbf{0.83} &\textbf{0.447}		&0.96	&0.89	&0.82 	&\textbf{0.88}	&\textbf{0.464}	 &0.93	&0.95	&0.86 	&\textbf{0.89} &\textbf{0.473}\\ 
	%&0.85	&0.54	&0.86 	&0.63	&0.977
		\hline
    
\end{tabular}
\end{center}
\end{table*}

\subsection{Image Category Retrieval} 
\label{sec:exp_cat}

For the image category retrieval, we consider a set of $42K$ images from imageNet database~\cite{deng2009imagenet} with $7$ synsets (categories) (namely \emph{animal, bottle, flower, furniture, geography, music, vehicle}) with five subtopics for each. Images are represented as a bag of visual words histogram with a vocabulary size of $48K$ over the densely extracted SIFT vectors. For each categorical query, we train an SVM hyperplane using LIBLINEAR~\cite{fan2008liblinear}. Since, there are only seven categories in our dataset, for each category we created 50 queries by randomly sampling 10\% of the images. After creating the queries, we are left with $35K$ images which we use for the retrieval task. We report the quantitative results in Table~\ref{table:imagenet_results} by the mean performance of all 350 queries. A few qualitative results on this dataset are shown in Figure~\ref{fig:imagenet_qresults}. 

We conducted two sets of experiments, 1) Retrieval without using hash functions and 2) Retrieval using hash functions, to evaluate the effectiveness of our proposed method. In the first set of experiments, we directly apply the existing diverse retrieval methods on the complete dataset. In the second set of experiments, we first select a candidate set of points by using the hash functions and then apply one of these methods to retrieve the images. 

We hypothesize that using hash functions in combination with any of the diverse retrieval methods will improve the diversity and the overall performance (h-score) with significant speed-ups. To validate our hypothesis, we evaluate various diverse retrieval methods in combination with our hash functions as described in Algortihm~\ref{algo:alg_lsh}. It can be noted that lines 6-10 in Algorithm~\ref{algo:alg_lsh} can be replaced with various retrieval methods and can be compared against the methods without hash functions. In particular, we show the comparison with the following retrieval methods: the k-nearest neighbor (NN), the QP-Rel method and the diverse retrieval methods like Backward selection (Rerank), Greedy~\cite{he2012gender}, MMR~\cite{carbonell1998use}. In Table~\ref{table:imagenet_results}, we denote NH as Null Hash i.e, without using any hash function, LSH-Div with the random hash function and LSH-SDiv with the (randomized) PCA hash function. 

We can see in Table~\ref{table:imagenet_results},  that our hash functions in combination with various methods are superior to the methods with NH. Our extensions based on LSH-Div and LSH-SDiv hash functions out-perform in all cases with respect to the h-score. Interestingly, LSH-Div and LSH-SDiv with NN report maximum h-score than any other methods. This observation implies that diversity can be preserved in the retrieval by directly using standard LSH based nearest neighbor method. We also report a significant speed up even for a moderate database of $35K$ images. Readers familiar with LSH will also agree that our methods will enjoy better speed up in presence of larger databases and higher dimensional representations. 

In Table~\ref{table:imagenet_results} the greedy method with our hash functions reports very low precision at top-20 and top-30 retrievals. This indicates that the greedy method may sometimes pick points too far from the query and might report images that are not relevant to the query. This observation is illustrated with our toy dataset in Figure~\ref{fig:toyDataset}. Notice that the existing diverse retrieval methods with NH report diverse images, but they are highly inefficient with respect to the retrieval time. Especially, the QP-Rel method also needs unreasonable memory for storing the gram matrix. To avoid any memory leaks, we partitioned the images into seven (number of categories) blocks and evaluated the queries independently i.e., when the query is flower, we only look at the block of flower images and retrieve diverse set of flowers. Although the QP-Rel method acheives better diversity, it is still computationally very expensive. Having such clear partitions is highly impractical and not feasible for other task/large datasets. We therefore, omit the results using QP-Rel method on the multi-label prediction task. 

%------------------------------------------------------------------------ 

\subsection{Multi-label Prediction} 
\label{sec:exp_lab}

We use one of the largest  multi-label datasets, LSHTC, to show the effectiveness of our proposed method. This dataset contains  the wikipedia documents with more than $300K$ labels. To avoid any bias towards the most frequently occurring labels, we selected only the documents which have at least 4 or more labels.  Thus, we have a data set of $754K$ documents with $259K$ unique labels. For our experiment, we randomly divide the data in 4:1 ratio for training and testing respectively. We use the  large scale multi label learning (LEML)~\cite{yu2014large} algorithm to train a linear multi-class classifier. This method is shown to provide state of the art results on many large multi label prediction tasks.
	
In Table~\ref{table:lshtc_results}, we report the performance of the label prediction with LEML and compare with our methods that predict diverse labels efficiently. Since, the number of labels for each document varies, we used a threshold parameter to limit the number of predicted labels to the documents. We selected the threshold by cross validating such  that it  maximizes the h-score. The precision and recall values corresponding to this setting are shown in the table. We also show the f-score computated as the harmonic mean of precision and recall in each case.

In LSHTC3 dataset, the labels are associated with a category hierarchy which is cyclic and unbalanced i.e., both the documents and subcategories are allowed to belong to more than one other category. In such cases, the notion of diversity i.e., the extent to which the predicted labels belong to multiple categores can be estimated using drank~\cite{SwamyRS14}. Since, the category hierarchy graph is cyclic, we prune the hierarchy graph to obtain a balanced tree by using the BFS traversal. The diversity of the predicted labels is computed as the drank score on this balanced tree. In Table~\ref{table:lshtc_results}, we report the overall performance of a method in terms of h-score i.e., the precision and the drank score.
	 
As can be seen from Table~\ref{table:lshtc_results}, the LSH-Div method shows a reasonable speedup but fails to report many of the accurate labels i.e., has low precision. Since, the LSHTC3 dataset is highly sparse in a large dimension, random projections generated by LSH-Div method are a bit inaccurate and might have resulted in poor accuracy. 

The proposed LSH-SDiv approach significantly boosts the accuracy, since, the random vectors in the hash function are projected onto the principal components that capture the data distribution accurately. The results shown in table are obtained by using 100 random projections for both LSH-Div and LSH-SDiv hash functions. For the LSH-SDiv method, we project the random projections onto the top 200 singular vectors obtained from the data points. 

Clearly, LSH-SDiv based hash function improves the diversity within the labels and outperforms LEML, MMR, PCA-Hash and LSH-Div methods in terms of overall performance (h-score). In summary, we obtain a speed-up greater than 20 over LEML method and greater than 80 over MMR method on this dataset. Note that, we omitted the results with greedy method as they failed to report accurate labels in this task.  

\begin{table}[t]
\begin{center}
\caption{Results on LSHTC3 challenge dataset with LEML, MMR, PCA-Hash, LSH-Div and LSH-SDiv methods. LSH-SDiv method significantly outperforms both LEML, MMR, PCA-Hash and LSH-Div methods in terms of overall performance, h-score as well as the retrieval time.}
\label{table:lshtc_results}
\begin{tabular}{|l|c|c|c|c|c|r|}
\hline
\textbf{Method} &  P & R & f-score & D & h & \pbox{10cm}{ time (msec)}\\
\hline\hline
LEML~\cite{yu2014large}         &0.304   &0.196      &0.192   &0.827        	&0.534 			&137.1\\ \hline
MMR~\cite{carbonell1998use}         &0.275   &0.134      &0.175   &0.865        	&0.418 			&458.8\\ \hline
PCA-Hash							 &0.265   &0.096	   &0.121   &0.872         &0.669          &5.9 \\ \hline
LSH-Div     					     &0.144   &0.088      &0.083   &0.825    		&0.437 			&7.2\\ \hline
LSH-SDiv     					     &0.318   &0.102      &0.133   &\textbf{0.919} 	&\textbf{0.734} &\textbf{5.7}\\
\hline
\end{tabular}
\end{center}
\vspace{-0.5cm}
\end{table}

\begin{figure*}
\begin{subfigure}{.33\textwidth}
  \centering
  \includegraphics[width=\linewidth]{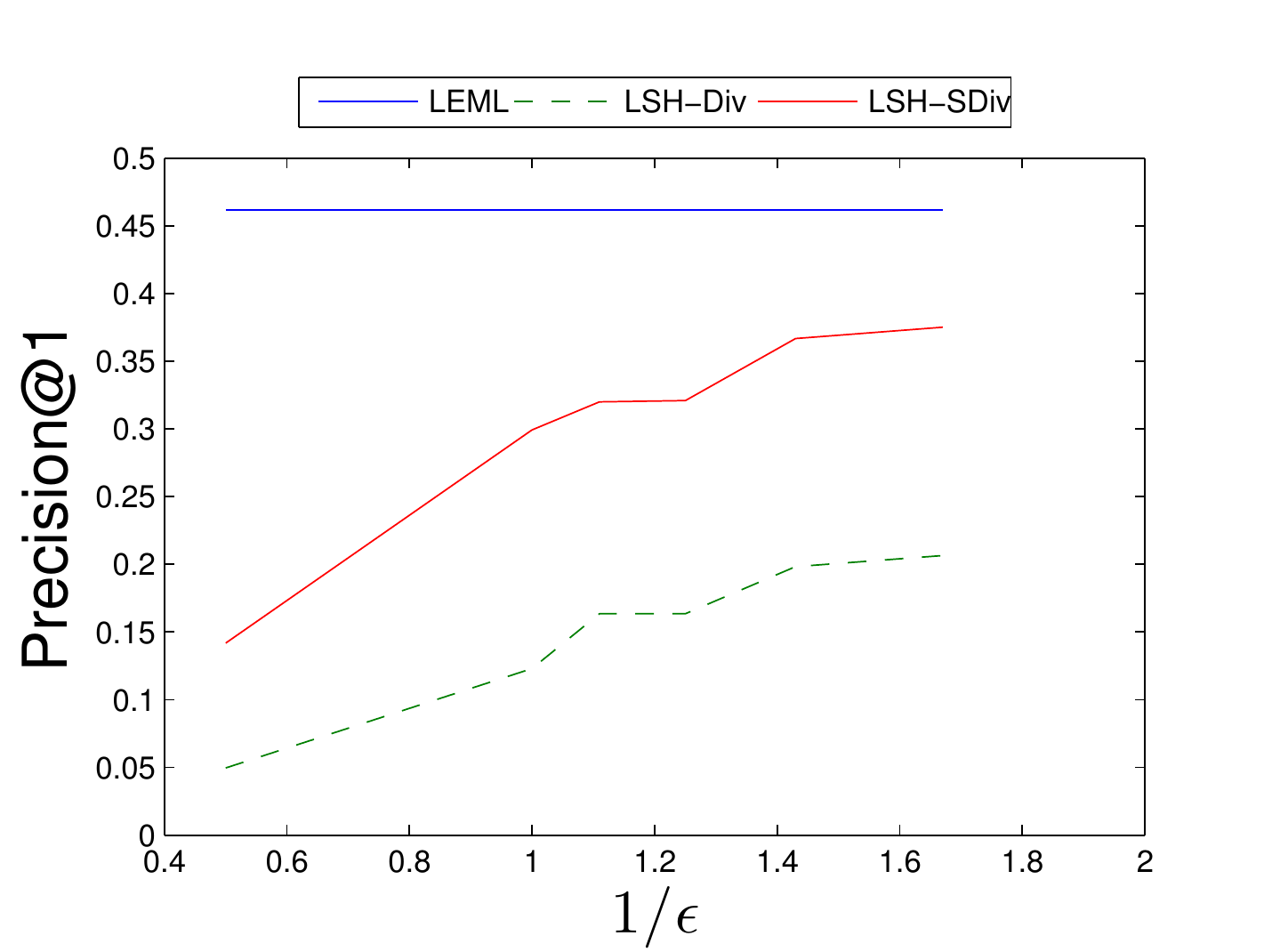}
  \caption{}
  \label{fig:example_naive}
\end{subfigure}%
\begin{subfigure}{.33\textwidth}
  \centering
  \includegraphics[width=\linewidth]{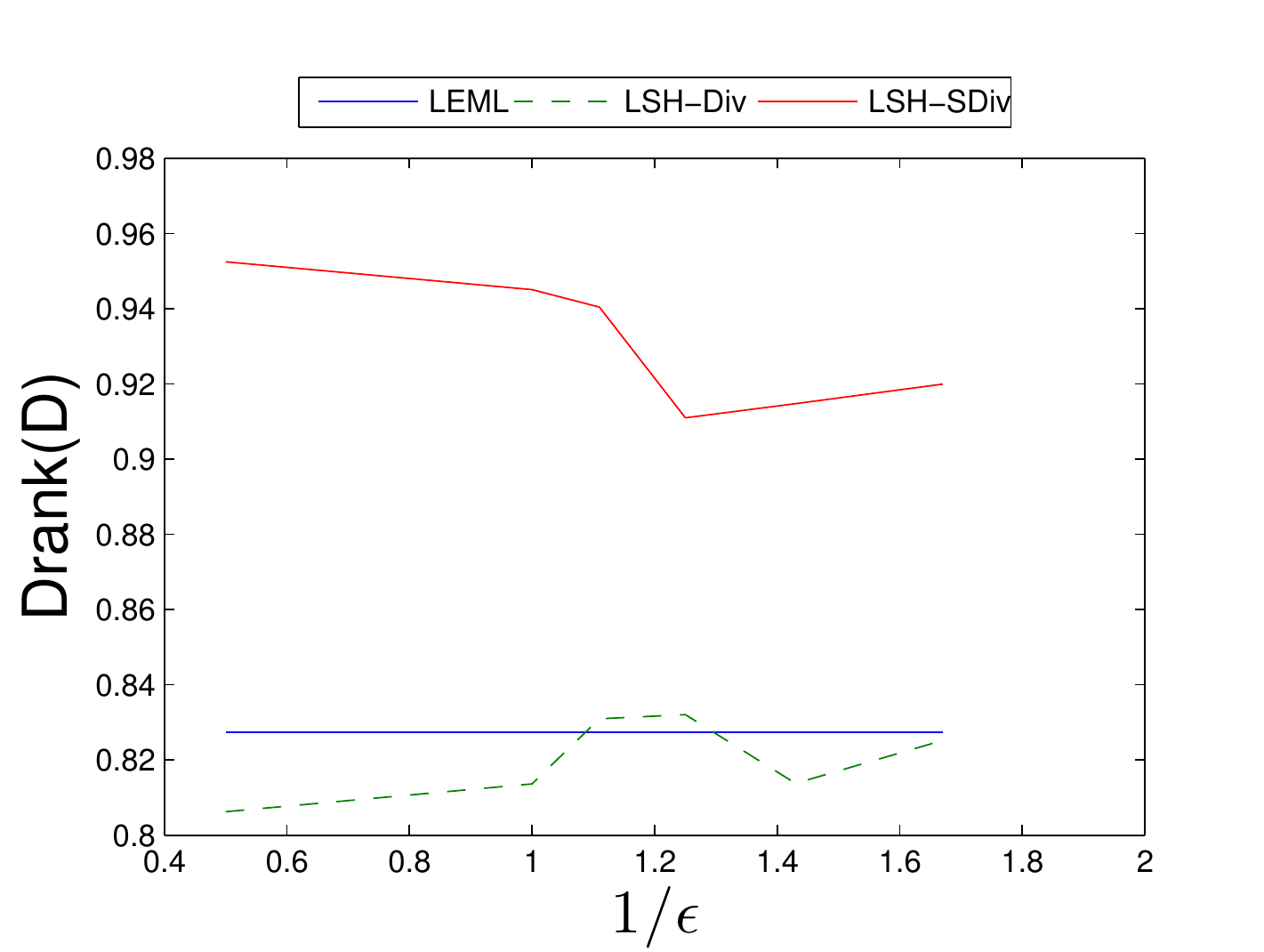}
 	 \caption{}
  \label{fig:example_lsh}
\end{subfigure}%
\begin{subfigure}{.33\textwidth}
  \centering
  \includegraphics[width=\linewidth]{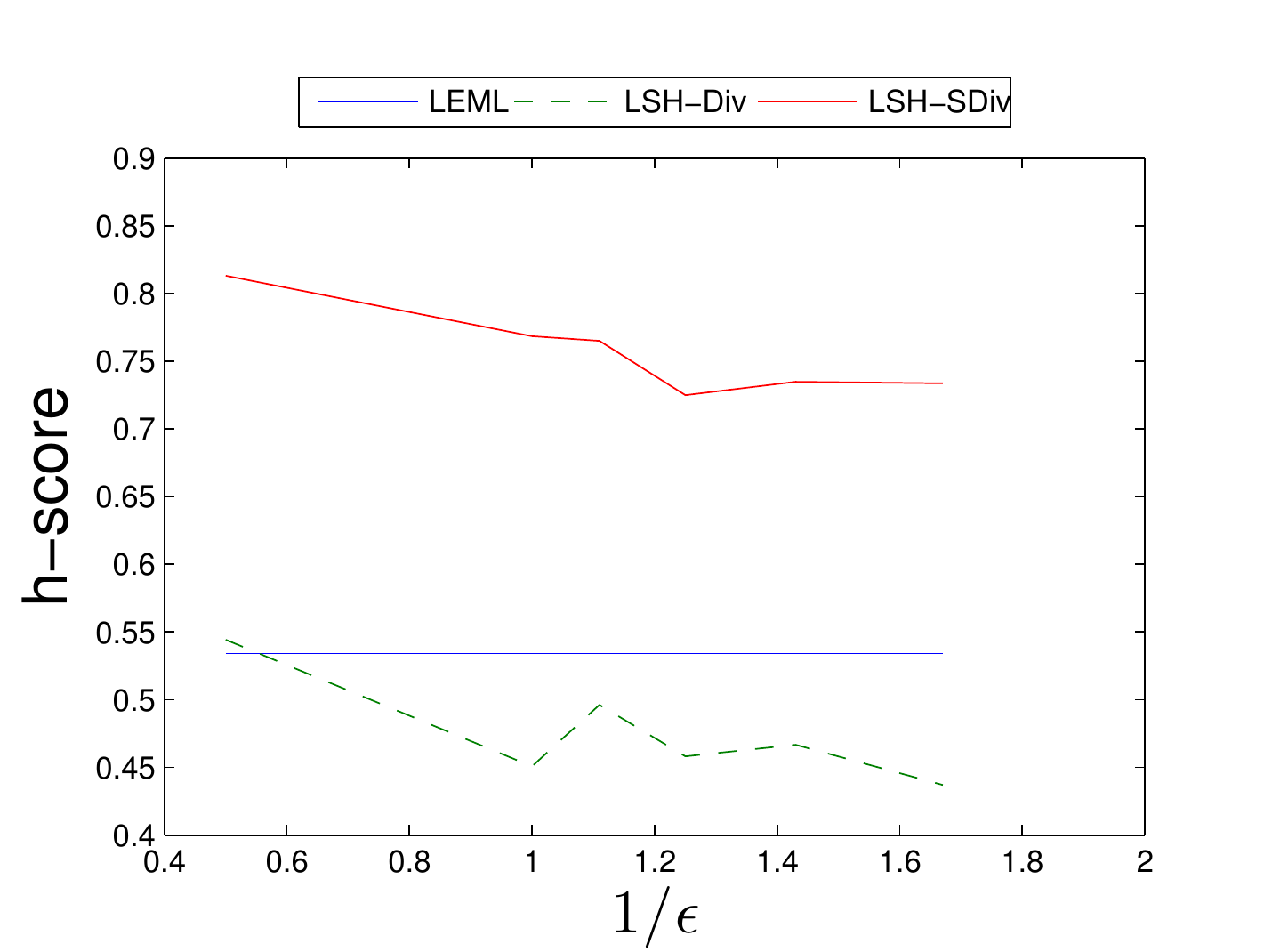}
   \caption{}
  \label{fig:example_lsh1}
\end{subfigure}
\caption{Results on LSHTC3. (a) LSH-SDiv method gives better precision than LSH-Div method. (b) LSH-SDiv method shows better diversity than LEML and LSH-Div methods. (c) LSH-SDiv method performs significantly better than the LEML and LSH-Div methods in terms of h-score. ({\em Image best viewed in color.})}
\vspace{-0.5cm}
\label{fig:lshVSvd_plot}
\end{figure*}

%%%

\subsection{Discussions}
\label{sec:discussion}

In this section, we focus on showing the trade-offs between accuracy, diversity and run-time. Figure~\ref{fig:lshVSvd_plot} illustrates the performance on LSHTC3 dataset with respect to the parameter $\epsilon$. In the figure we show the performance obtained when 100 random projections are selected for the LSH-Div method. For the LSH-SDiv method we project the 100 random projections onto the top-200 singular vectors obtained from the data. Notice that the conventional LSH hash function considers only random projections and fails to five good accuracy. As discussed in Section~\ref{sec:lsh_div}, a large number of random projections are needed to retrieve accurate labels, which would slow down the retrieval procedure.
 
In contrast, the LSH-SDiv method can successfully preserve the distances i.e., report accurate labels by projecting onto a set of $\beta$ principal components if the data is embedded in $\beta$ dimensions only. Similarly, if the $\beta+1$-th singular value of the data matrix is $\sigma_{\beta+1}$ then the distances are preserved upto that error and has no dependence on say $\epsilon$ that is required by standard LSH hash function. Hence, LSH-SDiv based technique typically requires much smaller number of hash functions than the standard LSH method and hence, is much faster as well (see Table~\ref{table:imagenet_results} and Table~\ref{table:lshtc_results}). 

Our empirical evidences from the above experiments confirm that we have a high precision for the image category retrieval scenario, a low precision for the multi-label prediction scenario. We demonstrated that the proposed algorithm is effective and robust, since, it improves diversity even when retrieving relevant results is difficult. Moreover, our algorithm can adopt to the data distribution while still retrieving accurate and diverse results. Our approach comes with an additional advantage of being more efficient computationally, which is crucial for large datasets.

\begin{figure*}
  \centering
  \includegraphics[width=\linewidth]{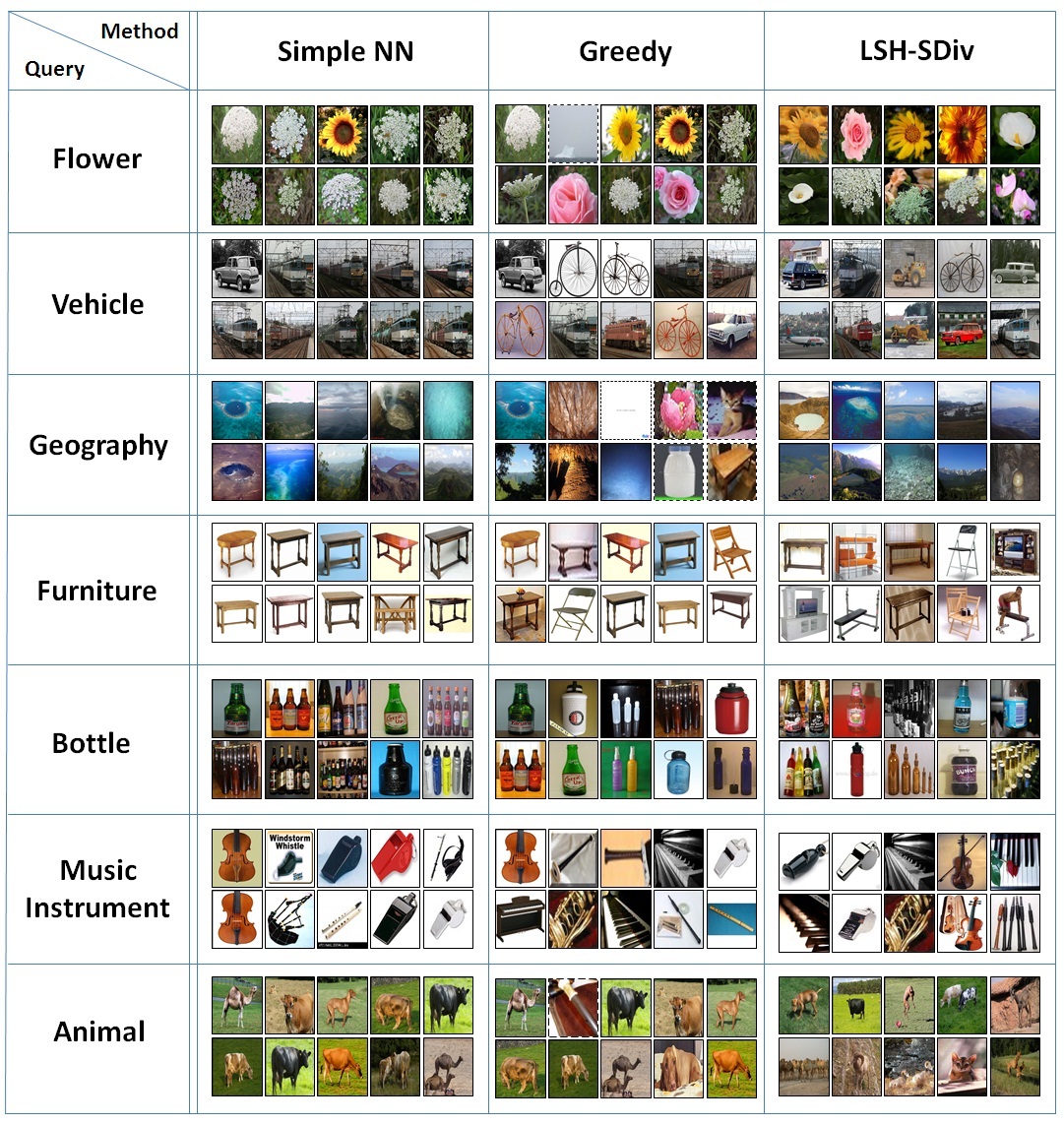}
\caption{In the plot, we show qualitative results for seven example queries from the ImageNet database. Top-10 retrieved images are shown for three methods: the first column with the simple NN method, the second coloum with Greedy MMR method, and the third coloumn with the proposed LSH-SDiv method. The images marked with dotted box are the incorrectly retrieved images with respect to the query. Notice that the greedy method fails to retrieve accurate retrieval for some of the queries. Our method, consistently retrieves relevant images and simultaneouly shows better diversity. ({\em Image best viewed in color.})}
\label{fig:imagenet_qresults}
\end{figure*}

%%% 

\section{Conclusions}
\label{sec:conc}
In this paper, we present an approach to efficiently retrieve diverse results based on {\em randomized} locality sensitive hashing. We argue that standard hash functions retrieve points that are sampled uniformly at random in all directions and hence ensure diversity by default. We show that, for two applications (image and text), our proposed methods retrieve significantly more diverse and accurate data points, when compared to the existing methods. The results obtained by our approach are appealing: a good balance between accuracy and diversity is obtained by using only a small number of hash functions. We obtained $100x$-speed-up over existing diverse retrieval methods while ensuring high diversity in retrieval. The proposed solution is an highly efficient with theoretical guarantees for the sub-linear retrieval time and therefore, the algorithms are interesting and should make more useful and attractive for all practical purposes.

\section{Future Work}

We believe that other approximate nearest neighbor retrieval algorithms like Randomized KD-Trees also encourage diversity in the retrieval. In our case, the rigorous theory of locality senstive hashing functions naturally supports its performance in relevance, diversity and retrieval time. Note that the random hash functions designed in our methods are only geared to maintain spread among points with very high probability. While doing so, the algorithm has no way of knowing which solutions are diverse and which are not diverse. Therefore, for these methods, the task of providing any guarantees of the true solution to the diverse retrieval problem is challenging. With this respect, it would be interesting to examine the existence of approximation guarantees to the optimal solution.

%%%

%ACKNOWLEDGMENTS are optional
%\section{Acknowledgments}
%This section is optional;

\bibliographystyle{abbrv}
\bibliography{sigproc} 

\balancecolumns

\end{document}